\documentclass[notitlepage,oneside,11pt]{amsart}
\usepackage{amsmath}
\usepackage{amssymb}
\usepackage[colorlinks=true,linkcolor=blue,citecolor=blue,urlcolor=magenta,pdfstartview=FitB,bookmarksopen=true]{hyperref}
\usepackage{natbib}
\usepackage[ top= 28mm,bottom=28mm ]{geometry}
\usepackage[onehalfspacing]{setspace}
\usepackage{amsfonts}
\usepackage{graphicx,float}

\newtheorem{theorem}{Theorem}
\theoremstyle{plain}

\newtheorem{corollary}{Corollary}

\newtheorem{lemma}{Lemma}

\newtheorem{proposition}{Proposition}

\numberwithin{equation}{section}

\newcommand{\rndbr}[1]{\left(#1\right)}

\newcommand{\clbr}[1]{\left\{#1\right\}}
\newcommand{\bm}{\mathbf}

\begin{document}

\title[Explicit solutions to two econometric models]{Explicit solutions to a vector time series model \\ and its induced model for business cycles}
\author{Xiongzhi Chen}
\address{Center for Statistics and Machine Learning and Lewis-Sigler Institute for Integrative Genomics\\
Princeton University\\
Princeton, NJ 08544, USA}
\email{xiongzhi@princeton.edu}

\begin{abstract}
This article gives the explicit solution to a general vector time series model that describes interacting,
heterogeneous agents that operate under uncertainties but according to
Keynesian principles, from which a model for business cycle is induced by a weighted average of the growth rates of the agents in the model. The explicit solution enables a direct simulation of the time series defined by the
model and better understanding of the joint behavior of the growth rates.
In addition, the induced model for business cycles and its solutions are explicitly given and analyzed.
The explicit solutions provide a better understanding of the mathematics of these models and the econometric properties they try to incorporate.
\end{abstract}
\maketitle
\tableofcontents

\section{Introduction\label{Sec:Intro}}

The study and modeling of business cycles have significant importance in
economic theory and practice; see, e.g., \cite{Burns:1946}, \cite{Lucas:1977},
\cite{Kydland:1977} and \cite{Kydland:1982}. Models of business cycles have
various forms as surveyed in \cite{Ahking:1988} and \cite{Lucas:1991}. Among
them, one was proposed (earlier but then published) in \cite{Ormerod:2001},
and it is based on interacting, heterogeneous agents that behave under
uncertainties about the future but according Keynesian principles. The model
is quite general, and the rationale and good performance of the model are
detailed in Chapter 9 of \cite{Ormerod:1998}.
Even though only partial solution of the model as ``integral equations'' has
 been provided in \cite{Ormerod:2001} or \cite{Ormerod:1998}, neither the mathematical properties of the model nor
how the individual agents' growth rates that are tied by the model should behave has been analyzed. Further, the induced
model for business cycles has not been analyzed mathematically. This makes
understanding of the long term econometric behavior of the growth rates employed in the model and of the
business cycles the induced model is able to capture, less
transparent and somewhat difficult. To resolve these issues, we derive the
explicit solutions to the model and the induced model, analyze the key
properties of these solutions, and make connections between their
mathematical features and econometric implications.

The rest of the article is organized as follows. In \autoref{Sec:Model} we
state the autoregressive model, and the explicit decomposition of the transition matrix (see (\ref{2}))
involved in the model is provided in \autoref{Sec:CanonicalForm}. The explicit solution of the model and its
properties are given in \autoref{Sec:Solution}, and solutions to the induced model for business cycles
and their properties are explored in \autoref{sec:busiCyc}. A brief discussion in
\autoref{Sec:Discussion} ends the article.

\section{The Vector Autoregressive Model\label{Sec:Model}}

For $i=1,\ldots,n$ with $n\in\mathbb{R}$ and $t\in\mathbb{Z}_{+}=\left\{
m\in\mathbb{Z}:m\geq0\right\}  $, let $x_{i}\left(  t\right)  $ be the growth
rate of the output of the $i$'th firm in period $t$ and $y_{i}\left(  t\right)  $
the rate of change of the sentiment about the future of the $i$'th firm formed
in period $t$. Further, define
\[
\mathcal{C}_{n}=\left\{  \mathbf{w}=\left(  w_{1},\ldots,w_{n}\right)
\in\mathbb{R}^{n}:\min_{1\leq i\leq n}w_{i}>0,\sum\nolimits_{i=1}^{n}
w_{i}=1\right\}  \text{.}
\]
The overall rate of growth of the output is the weighted sum of the individual
growth rates, defined as $\bar{x}\left(  t\right)  =\sum_{i=1}^{n}b_{i}
x_{i}\left(  t\right)  $ for some $\mathbf{b}=\left(  b_{1},\ldots,b_{n}\right)  \in
\mathcal{C}_{n}$, and the overall rate of growth of
sentiment is the weighted sum of the individual $y_{i}\left(  t\right)  $,
defined as $\bar{y}\left(  t\right)  =\sum_{i=1}^{n}a_{i}y_{i}\left(
t\right)  $ for some $\mathbf{a}=\left(  a_{1},\ldots
,a_{n}\right)  \in\mathcal{C}_{n}$. Further, $\left\{  x_{i}\left(  t\right)  \right\}
_{i=1}^{n}$ and $\left\{  y_{i}\left(  t\right)  \right\}  _{i=1}^{n}$ are
related by the model proposed in \cite{Ormerod:2001} as%
\begin{equation}
\left\{
\begin{array}
[c]{c}%
x_{i}\left(  t+1\right)  =\left(  1-\alpha\right)  x_{i}\left(  t\right)
+\alpha\left[  \bar{y}\left(  t\right)  +\varepsilon_{i}\left(  t\right)
\right] \\
y_{i}\left(  t+1\right)  =\left(  1-\beta\right)  y_{i}\left(  t\right)
-\beta\left[  \bar{x}\left(  t\right)  +\eta_{i}\left(  t\right)  \right]
\end{array}
\right.  \label{1}
\end{equation}
for constants $\alpha,\beta\in\mathbb{R}$, $\varepsilon_{i}\rndbr{t}$ $\propto
\mathsf{N}\left(  \mu_{i},\sigma_{i}^{2}\right)  $ and $\eta_{i}\rndbr{t}
=\varepsilon_{n+i}\rndbr{t}\propto\mathsf{N}\left(  \mu_{n+i},\sigma_{n+i}^{2}\right)
$ for $\mu_{i},\mu_{n+i}\in\mathbb{R}$ and $\sigma_{i}$ and $\sigma_{n+i}>0$,
where $\xi\propto\mathsf{N}\left(  \mu,\sigma^{2}\right)  $ means that $\xi$
is Normally distributed and has density
\[
g_{\mu,\sigma^{2}}\left(  x\right)  =\left(  \sqrt{2\pi}\sigma\right)
^{-1}\exp\left[  -\left(  2\sigma^{2}\right)  ^{-1}\left(  x-\mu\right)
^{2}\right]
\]
for $x\in\mathbb{R}$. Here we assume that all random vectors are defined on a
common probability space $\left(  \Omega,\mathcal{F},\mathbb{P}\right)  $,
where $\Omega$ is the sample space, $\mathcal{F}$ a sigma algebra on $\Omega$,
and $\mathbb{P}$ the probability measure.

Let $\boldsymbol{x}_{t}=\left(  x_{1}\left(  t\right)  ,\ldots,x_{n}\left(
t\right)  \right)  $, $\boldsymbol{y}_{t}=\left(  y_{1}\left(  t\right)
,\ldots,y_{n}\left(  t\right)  \right)  $, $\boldsymbol{z}_{t}=\left(
\boldsymbol{x}_{t},\boldsymbol{y}_{t}\right)  ^{T}$, $\boldsymbol{\varepsilon}
_{t}=\left(  \varepsilon_{1}\left(  t\right)  ,\ldots,\varepsilon_{n}\left(
t\right)  \right)  $,
$\boldsymbol{\eta}_{t}=\left(  \eta_{1}\left(  t\right)
,\ldots,\eta_{n}\left(  t\right)  \right)  $, and $\boldsymbol{\gamma
}_{t}=\left(  \alpha\boldsymbol{\varepsilon}\left(  t\right)  ,-\beta
\boldsymbol{\eta}\left(  t\right)  \right)  ^{T}$, where the superscript $T$
denote transpose of a matrix. Further, let$\ \mathcal{M}_{s\times s^{\prime}}$
with $s,s^{\prime}\in\mathbb{N}$ be the set of $s\times s^{\prime}$ real
matrices, which is denoted by $\mathcal{M}_{s}$ when $s=s^{\prime}$. Then model
(\ref{1}) can be rewritten as
\begin{equation}
\boldsymbol{z}_{t+1}=\mathbf{M}\boldsymbol{z}_{t}+\boldsymbol{\gamma}_{t},
\label{3}
\end{equation}
where the ``transition matrix''
\begin{equation}
\mathbf{M}=\left(
\begin{array}
[c]{cc}
\left(  1-\alpha\right)  \mathbf{I}_{n} & \alpha\mathbf{1}_{n}^{T}\mathbf{a}\\
-\beta\mathbf{1}_{n}^{T}\mathbf{b} & \left(  1-\beta\right)  \mathbf{I}_{n}
\end{array}
\right)  \in\mathcal{M}_{2n}, \label{2}
\end{equation}
$\mathbf{I}_{s}$ denotes the $s\times s$ identity matrix, and $\mathbf{1}_{s}$ is a row vector of $s$ one's.

Model \eqref{3} bundles the agents' growth rates stored in the vector $\boldsymbol{z}_{t}$ together,
in which $\boldsymbol{z}_{t}$ is sent into the very near future by the mapping induced by the matrix $\mathbf{M}$.
Therefore, it put restrictions on how $\boldsymbol{z}_{t}$
should behave jointly. However, no explicit solution in $\clbr{\boldsymbol{z}_{t}}_{\mathbb{Z}_{+}}$  has been available.
This makes hard the direct simulation of $\clbr{\boldsymbol{z}_{t}}_{\mathbb{Z}_{+}}$ and difficult the understanding of the long term behavior of $\clbr{\boldsymbol{z}_{t}}_{\mathbb{Z}_{+}}$ both mathematically and
econometrically.

\section{Decomposition of The Transition Matrix \label{Sec:CanonicalForm}}

To probe into the long term behavior of $\clbr{\boldsymbol{z}_{t}}_{t\in \mathbb{Z}_{+}}$ and how they are tied together in  model (\ref{3}), an
efficient strategy is to decompose the $2n\times2n$ matrix $\mathbf{M}$ into products of simpler matrices. To this end,
we need to understand the nontrivial invariant subspaces, if they exist, of the mapping induced by the $2n\times2n$
matrix $\mathbf{M}$. To maintain good economic meaning of model \eqref{3}, it is natural to assume
\begin{equation}
\rndbr{\alpha,\beta} \notin \clbr{\rndbr{0,0},\rndbr{1,1}}. \label{eq:rst}
\end{equation}
The results to be presented in this section are independent of the distributional assumptions on $\boldsymbol{\gamma}_t$
for $t \in \mathbb{Z}_{+}$.

\subsection{Jordan Canonical Form of $\bm{M}$}\label{subsec:Jform}

In this subsection, we provide the Jordan canonical form (see, e.g., \cite{Jacobson:1953} for a definition) of $\mathbf{M}$
in the vector space $\mathbb{R}^{2n}$ over $\mathbb{R}$. This will help convert the iterative identity, i.e., (\ref{3}), for $\clbr{\boldsymbol{z}_{t}}_{t\in \mathbb{Z}_{+}}$ into a direct, explicit representation in \eqref{28} without computing the averages $\bar{x}_t$ or $\bar{y}_t$. For $\theta\in \mathbb{R\ }$set $\mathbf{J}_{1}\left(  \theta\right)  =\theta$, and for a natural number $r\geq2$ define the Jordan block
\begin{equation}
\mathbf{J}_{r}\left(  \theta\right)  =\left(
\begin{array}
[c]{ccccc}%
\theta & 1 &  &  & \\
& \theta & 1 &  & \\
&  & \ddots & \ddots & \\
&  &  & \theta & 1\\
&  &  &  & \theta
\end{array}
\right)  \in\mathcal{M}_{r} \label{4'}
\end{equation}
whose diagonal entries are all $\theta$, superdiagonal entries are all $1$,
and unmarked entries are identically zero. Let
$f\left(  \lambda\right)  =\left\vert \lambda\mathbf{I}-\mathbf{M}\right\vert$ be the characteristic polynomial of $\mathbf{M}$, $\Delta = \alpha^2 + \beta^2 - 6 \alpha \beta$, $d_1 = \rndbr{3 - 2\sqrt{2}}\beta$ and $d_2 = \rndbr{3+2\sqrt{2}}\beta$. The following theorem gives the roots of $f\left(  \lambda\right)$ and conditions on if $\mathbf{M}$ can be diagonalized.

\begin{theorem}
\label{Thm:JordanForm}The characteristic polynomial
\begin{equation}
f\left(  \lambda\right)=\left(  \lambda-1+\beta\right)  ^{n-1}\left(  \lambda-1+\alpha\right)
^{n-1}g\rndbr{\lambda} \label{8a}
\end{equation}
with
\begin{equation}
g\rndbr{\lambda} = \left(  \lambda-1\right)^2 + \left(  \lambda-1\right)\left(\alpha + \beta\right)
+2\alpha\beta.\label{8b}
\end{equation}
So, $f\left(  \lambda\right)$ always has real roots $\lambda_1 = 1 - \alpha$ and $\lambda_2 = 1 - \beta$.
In addition, the following hold:
\begin{enumerate}
  \item If $\min\clbr{d_1,d_2} < \alpha < \max\clbr{d_1,d_2}$, then $f\rndbr{\lambda}$ has no other real roots
       and $\mathbf{M}$ can not be diagonalized in the vector space $\mathbb{R}^{2n}$ over $\mathbb{R}$.

  \item If $\alpha \leq \min\clbr{d_1,d_2}$ or $\alpha \geq \max\clbr{d_1,d_2}$, then $f\left(  \lambda\right)$ has two more real roots $\lambda_3 = 2^{-1}\rndbr{2-\alpha -\beta + \sqrt{\Delta}}$ and $\lambda_4 = 2^{-1}\rndbr{2-\alpha -\beta - \sqrt{\Delta}}$. If $\alpha < \min\clbr{d_1,d_2}$ or $\alpha > \max\clbr{d_1,d_2}$, the Jordan canonical form of $\mathbf{M}$ is the diagonal matrix
   \begin{equation}
    \mathbf{J=Q}^{-1}\mathbf{MQ}=\mathsf{diag}\left\{  \lambda_{1}\mathbf{I}
    _{n-1},\lambda_{3}\mathbf{I}_{1},\lambda_{2}\mathbf{I}_{n-1},\lambda_{4}\mathbf{I}_{1}\right\}  \label{4}
    \end{equation}
  for some nonsingular matrix $\mathbf{Q}\in\mathcal{M}_{2n}$.
  However, if $\alpha = d_1$ or $\alpha = d_2$, $\mathbf{M}$ can not be diagonalized in the vector space $\mathbb{R}^{2n}$ over $\mathbb{R}$ and the Jordan canonical form of $\mathbf{M}$ is
  \begin{equation}
  \mathbf{J=Q}^{-1}\mathbf{MQ}=\mathsf{diag}\left\{ \lambda _{1}\mathbf{I}
  _{n-1},\lambda _{2}\mathbf{I}_{n-1},\mathbf{J}_{2}\left( \lambda _{3}\right) \right\}\label{4a}
  \end{equation}
  for some nonsingular matrix $\mathbf{Q}\in\mathcal{M}_{2n}$.
\end{enumerate}

\end{theorem}

\begin{proof}
For some $\varepsilon \geq 0$, let
\begin{equation}
\mathbf{M}_{\varepsilon}=\left(
\begin{array}
[c]{cc}
\left( 1-\alpha -\varepsilon\right)  \mathbf{I}_{n} & \alpha\mathbf{1}_{n}^{T}\mathbf{a}\\
-\beta\mathbf{1}_{n}^{T}\mathbf{b} & \left(  1-\beta\right)\mathbf{I}_{n}
\end{array}
\right)  , \label{7}
\end{equation}
and
\begin{equation}
\mathbf{M}_{\lambda,\varepsilon}= \lambda\mathbf{I}_{2n}-\mathbf{M}_{\varepsilon}=\left(
\begin{array}
[c]{cc}
\left(  \lambda-1+\alpha+\varepsilon\right)  \mathbf{I}_{n} & -\alpha\mathbf{1}_{n}^{T}\mathbf{a}\\
\beta\mathbf{1}_{n}^{T}\mathbf{b} & \left(  \lambda-1+\beta\right)
\mathbf{I}_{n}
\end{array}
\right). \label{7}
\end{equation}
Then $\mathbf{M}_{0} = \mathbf{M}$, and it's obvious that $\lambda_{1}=1-\alpha$ and $\lambda_{2}=1-\beta$ are roots
of $f\left(  \lambda\right)$ since $\mathsf{rank}\left(  \mathbf{M}_{\lambda,0}\right)  =n+1<2n$
when $\lambda = \lambda_1$ or $\lambda = \lambda_2$. To find other roots of $f\left(  \lambda\right)$, set
\[
\mathbf{T}_{\varepsilon}=\left(
\begin{array}
[c]{cc}
\mathbf{I}_{n} & \mathbf{0}\\
\frac{-\beta}{\lambda-1+\alpha + \varepsilon}\mathbf{1}_{n}^{T}\mathbf{b} & \mathbf{I}_{n}
\end{array}
\right)  \in\mathcal{M}_{2n}
\]
where $\varepsilon >0$ is now assumed and $\mathbf{0}$ denotes a matrix with identical zero entries of compatible
dimension. Then $\left\vert \mathbf{T}_{\varepsilon}\right\vert =1$ and
\begin{align}
\mathbf{\tilde{M}}_{\varepsilon}  &  =\mathbf{T}_{\varepsilon}\mathbf{M}_{\lambda,\varepsilon}\nonumber\\
&  =\left(
\begin{array}
[c]{cc}
\mathbf{I}_{n} & \mathbf{0}\nonumber\\
\frac{-\beta}{\lambda-1+\alpha+\varepsilon}\mathbf{1}_{n}^{T}\mathbf{b} & \mathbf{I}_{n}
\end{array}
\right)  \left(
\begin{array}
[c]{cc}
\left(  \lambda-1+\alpha+\varepsilon\right)  \mathbf{I} & -\alpha\mathbf{1}^{T}\mathbf{a}\\
\beta\mathbf{1}^{T}\mathbf{b} & \left(  \lambda-1+\beta\right)\mathbf{I}_{n}
\end{array}
\right) \\
&  =\left(
\begin{array}
[c]{cc}
\left(  \lambda-1+\alpha+\varepsilon\right)  \mathbf{I}_{n} & -\alpha\mathbf{1}_{n}^{T}\mathbf{a}\\
\mathbf{0} & \frac{\alpha\beta}{\lambda-1+\alpha+\varepsilon}\mathbf{b1}_{n}
^{T}\mathbf{1}_{n}^{T}\mathbf{a}+\left(  \lambda-1+\beta\right)\mathbf{I}_{n}
\end{array}
\right)  \text{.}\label{eq:8c}
\end{align}
Consequently, using the fact $\mathbf{b1}_{n}^{T}=1$ and Sylvester's
determinant theorem, we obtain
\begin{align}
f\left(  \lambda\right)   &  =\left\vert \lambda\mathbf{I}_{2n}-\mathbf{M}
\right\vert =\lim_{\varepsilon \rightarrow 0}\left\vert \mathbf{T}_{\varepsilon}^{-1}\right\vert
\left\vert \mathbf{\tilde{M}}_{\varepsilon}\right\vert =\lim_{\varepsilon \rightarrow 0}\left\vert \mathbf{\tilde{M}}_{\varepsilon}\right\vert \nonumber\\
&  =\lim_{\varepsilon \rightarrow 0}\left\vert \left(  \lambda-1+\alpha+\varepsilon\right)  \mathbf{I}\right\vert
_{n}\left\vert \frac{\alpha\beta}{\lambda-1+\alpha+\varepsilon}\mathbf{1}_{n}
^{T}\mathbf{a}+\left(  \lambda-1+\beta\right)  \mathbf{I}_{n}\right\vert \nonumber\\
&  =\lim_{\varepsilon \rightarrow 0}\left(  \lambda-1+\alpha+\varepsilon\right)  ^{n}\left(  \lambda-1+\beta\right)
^{n-1}\left[  \left(  \lambda-1+\beta\right)  +\frac{\alpha\beta}
{\lambda-1+\alpha+\varepsilon}\right] \nonumber
\end{align}
Thus, \eqref{8a} and \eqref{8b} hold, i.e.,
\[
f\rndbr{\lambda}=\left(  \lambda-1+\beta\right)  ^{n-1}\left(  \lambda-1+\alpha\right)^{n-1}g\rndbr{\lambda}.
\]
This means that $f\rndbr{\lambda}$ always has roots $\lambda_1 = 1 - \alpha$ and $\lambda_2 = 1 - \beta$.

Now we deal with the extra roots of $f\rndbr{\lambda}$, for which the theory in \cite{Jacobson:1953} on Jordan canonical form will be applied. Without loss of generality (WLOG), assume for the rest of the proof that $d_1=\min\clbr{d_1,d_2}$ and $d_2=\max\clbr{d_1,d_2}$. It is easy to verify that the determinant of $g$ in \eqref{8b} is $\Delta = \alpha^2 + \beta^2 - 6 \alpha \beta$ and that $\Delta=0$ if and only if $\alpha = \rndbr{3 - 2\sqrt{2}}\beta$ or $\alpha = \rndbr{3+2\sqrt{2}}\beta$. By \eqref{8a} and properties of quadratic functions, we see that $f$ has two other real roots $\lambda_{3}$ and $\lambda_{4}$ when $\alpha \leq d_1$ or $\alpha \geq d_2$ but no more real roots when $d_1 < \alpha < d_2$. Since $f\rndbr{\lambda}$ can not be written as $\prod_j \rndbr{\lambda - \lambda_j}$ for reals $\lambda_j$ when $d_1 < \alpha < d_2$, $\mathbf{M}$ can not be diagonalized in the vector space $\mathbb{R}^{2n}$ over $\mathbb{R}$.

Finally, we derive the Jordan blocks corresponding to each $\lambda_{i}$ for $i=1,\ldots,4$ when $\alpha \leq d_1$ or $\alpha \geq d_2$. In this case, $\alpha \neq \beta$, $f\rndbr{\lambda}$ has the form $\prod_j \rndbr{\lambda - \lambda_j}$ for real $\lambda_j$ and $\mathbf{M}$ can potentially be diagonalized. For $\lambda_{1}=1-\alpha$, we have
\begin{equation}
\mathbf{M}-\lambda_{1}\mathbf{I}_{2n}=\left(
\begin{array}
[c]{cc}
\mathbf{0} & \alpha\mathbf{1}_{n}^{T}\mathbf{a}\\
-\beta\mathbf{1}_{n}^{T}\mathbf{b} & \left(  \alpha-\beta\right)
\mathbf{I}_{n}
\end{array}
\right). \label{9}
\end{equation}
So, when $\alpha \neq \beta$, we have $\mathsf{rank}\left(  \mathbf{M}-\lambda_{1}\mathbf{I}_{2n}\right)  =n+1$ and
\[
\rho_{\mathbf{M}}\left(  \lambda_{1}\right)  =2n-\mathsf{rank}\left(
\mathbf{M}-\lambda_{1}\mathbf{I}_{2n}\right)  =n-1,
\]
where $\rho_{\mathbf{A}}\left(  \lambda\right)  $ denotes the dimension of the
kernel space of $\mathbf{A}-\lambda\mathbf{I}_{s}$ for a square matrix
$\mathbf{A}\in\mathcal{M}_{s}$ and $\lambda\in\mathbb{R}$ as a linear mapping
$\boldsymbol{v}\mapsto\left(  \mathbf{A}-\lambda\mathbf{I}_{s}\right)
\boldsymbol{v}$ for a column vector $\boldsymbol{v}\in\mathbb{R}^{s}$.

For $\lambda _{2}=1-\beta $, we have
\begin{equation}
\mathbf{M}-\lambda _{2}\mathbf{I}_{2n}=\left(
\begin{array}{cc}
\left( \beta -\alpha \right) \mathbf{I}_{n} & \alpha \mathbf{1}_{n}^{T}\mathbf{a} \\
-\beta \mathbf{1}_{n}^{T}\mathbf{b} & \mathbf{0}
\end{array}
\right).   \label{11}
\end{equation}
So, when $\alpha \neq \beta$, $\mathsf{rank}\left( \mathbf{M}-\lambda _{2}\mathbf{I}_{2n}\right) =n+1$
and
\begin{equation*}
\rho _{\mathbf{M}}\left( \lambda _{2}\right) =2n-\mathsf{rank}\left( \mathbf{
M-}\lambda _{2}\mathbf{I}_{2n}\right) =n-1.
\end{equation*}

For $\lambda_{3}$ and $\lambda_{4}$ when $\alpha < d_1$ or $\alpha > d_2$,
we immediately see that the Jordan blocks corresponding to them are respectively
$\mathbf{J}_{1}\left(  \lambda_{3}\right)  =\lambda_{3}$ and $\mathbf{J}
_{1}\left(  \lambda_{4}\right)  =\lambda_{4}$ since each of $\lambda_{3}$ and
$\lambda_{4}$ is a simple root of $f\left(  \lambda\right)$. Thus, $\sum_{i=1}^{4}\rho _{\mathbf{M}}\left( \lambda _{i}\right) =2n$ and there
is a nonsingular matrix $\mathbf{Q}$ such that (\ref{4}) holds.

However, for $\lambda_{3}$ and $\lambda_{4}$ when $\alpha = d_1$ or $\alpha = d_2$,
$\lambda_{3}=1-2^{-1}\left( \alpha +\beta \right)$ becomes a double root and
\begin{equation*}
\mathbf{M}-\lambda _{3}\mathbf{I}_{2n}=\left(
\begin{array}{cc}
2^{-1}\left( \beta -\alpha \right) \mathbf{I}_{n} & \alpha \mathbf{1}_{n}^{T}\mathbf{a} \\
-\beta \mathbf{1}_{n}^{T}\mathbf{b} & -2^{-1}\left( \beta -\alpha \right) \mathbf{I}_{n}
\end{array}
\right).
\end{equation*}
In order to decide $\rho _{\mathbf{M}}\left( \lambda _{3}\right)$, the rank $r_{\lambda_3}$ of $\mathbf{M}-\lambda _{3}\mathbf{I}_{2n}$ needs to be obtained. From \eqref{eq:8c}, we know that $r_{\lambda_3}$ is that of
\begin{equation*}
\mathbf{M}_{\lambda _{3}}=\left(
\begin{array}{cc}
\frac{\alpha -\beta }{2}\mathbf{I}_{n} & -\alpha \mathbf{1}_{n}^{T}\mathbf{a}
\\
\mathbf{0} & \frac{\beta -\alpha }{2}\mathbf{I}_{n}-\frac{2\alpha \beta }{
\beta -\alpha }\mathbf{1}_{n}^{T}\mathbf{a}
\end{array}
\right).
\end{equation*}
Set $\mathbf{B}=\frac{\beta -\alpha }{2}\mathbf{I}_{n}-\frac{2\alpha \beta
}{\beta -\alpha }\mathbf{1}_{n}^{T}\mathbf{a}$. Then $\left\vert \mathbf{B}\right\vert = 0$,
i.e., $\mathsf{rank}\left(\mathbf{B}\right) <n$ since $\mathbf{a}\mathbf{1}_{n}^{T}=1$
 and $\alpha =d_{1}$ or $d_{2}$ implies $\frac{\beta -\alpha }{2}=\frac{2\alpha \beta
}{\beta -\alpha }$.  So, it suffices to
get the rank of $\mathbf{B}$ to obtain $r_{\lambda_3}$. Let $\mathbf{a}_{\left( -1\right) }$ be the
vector obtained by removing one entry from $\mathbf{a}$ and $\mathbf{B}%
_{n-1}=\frac{\beta -\alpha }{2}\mathbf{I}_{n-1}-\frac{2\alpha \beta }{\beta
-\alpha }\mathbf{1}_{n-1}^{T}\mathbf{a}_{\left( -1\right) }$. Then
\begin{eqnarray*}
\left\vert \mathbf{B}_{n-1}\right\vert  &=&\left\vert \frac{\beta -\alpha }{2%
}\left( \mathbf{I}_{n-1}-\frac{4\alpha \beta }{\left( \beta -\alpha \right)
^{2}}\mathbf{1}_{n-1}^{T}\mathbf{a}_{\left( -1\right) }\right) \right\vert
\\
&=&\left( \frac{\beta -\alpha }{2}\right) ^{n-1}\left( 1-\mathbf{a}_{\left(
-1\right) }\mathbf{1}_{n-1}^{T}\right) \neq 0
\end{eqnarray*}
by the definition of $\mathbf{a}$. Therefore, $\mathsf{rank}\left( \mathbf{B}
\right) =n-1$ and $\rho _{\mathbf{M}}\left( \lambda _{3}\right) =2n-\mathsf{
rank}\left( \mathbf{M-}\lambda _{3}\mathbf{I}_{2n}\right) =1$. This implies
that the Jordan block corresponding to $\lambda_3$ is $\mathbf{J}_{2}\left( \lambda _{3}\right)$ and
that $\sum_{i=1}^{3}\rho _{\mathbf{M}}\left( \lambda _{i}\right) =2n-1$. Therefore,
$\mathbf{M}$ can not be diagonalized in the vector space $\mathbb{R}^{2n}$ over $\mathbb{R}$.
However, there exists a nonsingular $\mathbf{Q}\in \mathcal{M}_{2n}$ such that
\begin{equation*}
\mathbf{J=Q}^{-1}\mathbf{MQ}=\mathsf{diag}\left\{ \lambda _{1}\mathbf{I}
_{n-1},\lambda _{2}\mathbf{I}_{n-1},\mathbf{J}_{2}\left( \lambda _{3}\right)
\right\},
\end{equation*}
which justifies \eqref{4a}. This completes the proof.
\end{proof}

\subsection{Explicit Form of The Matrix of Basis}

\autoref{Thm:JordanForm} provides an eigen decomposition of $\mathbf{M}$. However, it does not show what the matrix of basis
$\mathbf{Q}$ is. In what follows, we will only provide explicitly $\mathbf{Q}$ for the second case in \autoref{Thm:JordanForm} for which $\bm{M}$ can be diagonalized, since this case makes $\clbr{\boldsymbol{z}_{t}}_{t\in \mathbb{Z}_{+}}$ the most amenable to an econometric analysis of its long term behavior. For an integer $s>1$, let $\mathbf{e}_{i,s}\in\mathbb{R}^{s}$ be such that only
the $i$th entry of $\mathbf{e}_{i,s}$ is $1$ but others are all zero, and for
$\mathcal{B}\subseteq\mathbb{R}^{s}$ let $\mathsf{span}\left(  \mathcal{B}\right)  $ be the smallest linear space containing $\mathcal{B}$. We have:

\begin{theorem}
\label{Thm:Basis} Suppose $\alpha < d_1$ or $\alpha > d_2$ such that \eqref{4} holds, then the matrix $\mathbf{Q}$ in \eqref{4} is given by
\begin{equation}
\mathbf{Q=}\left(
\begin{array}
[c]{cccc}
\mathbf{W}_{\lambda_{1}}^{T} & \mathbf{W}_{\lambda_{3}}^{T} & \mathbf{W}_{\lambda_{2}}^{T} & \mathbf{W}_{\lambda_{4}}^{T}
\end{array}
\right)  , \label{14}
\end{equation}
where, for $1\leq i\leq n-1$,
\begin{equation}
\left\{
\begin{array}{l}
\mathbf{W}_{\lambda _{1}}^{T}=\left( \boldsymbol{\varepsilon }
_{1}^{T},\ldots ,\boldsymbol{\varepsilon }_{n-1}^{T}\right) \text{ with }
\boldsymbol{\varepsilon }_{i}=\left( -b_{1}^{-1}b_{i+1},\mathbf{e}_{i,n-1},
\mathbf{0}\right) \\
\mathbf{W}_{\lambda _{2}}^{T}=\left( \boldsymbol{\tilde{\varepsilon}}
_{1}^{T},\ldots ,\boldsymbol{\tilde{\varepsilon}}_{n-1}^{T}\right) \text{
with }\boldsymbol{\tilde{\varepsilon}}_{i}=\left( \mathbf{0}
,-a_{1}^{-1}a_{i+1},\mathbf{e}_{i,n-1}\right) \\
\mathbf{W}_{\lambda _{3}}^{T}=\left( \mathbf{1}_{n},-2\alpha \left(
\beta -\alpha -\sqrt{\Delta }\right)^{-1} \mathbf{1}_{n}\right) ^{T} \\
\mathbf{W}_{\lambda _{4}}^{T}=\left( \mathbf{1}_{n},-2\alpha \left(
\beta -\alpha +\sqrt{\Delta }\right)^{-1} \mathbf{1}_{n}\right) ^{T}.
\end{array}
\right.   \label{15}
\end{equation}

\end{theorem}

\begin{proof}
Recall that $\mathbf{MQ=QJ}$, where%
\[
\mathbf{J}=\mathsf{diag}\left\{  \lambda_{1},\ldots,\lambda_{1},\lambda_{3},\lambda
_{2},\ldots,\lambda_{2},\lambda_{4}\right\}
\]
as in (\ref{4}), and $\lambda _{1}=1-\alpha $, $\lambda _{2}=1-\beta $,
$\lambda _{3,4}=2^{-1}\left( 2-\alpha -\beta \pm \sqrt{\Delta }\right) $. We will
find $\mathbf{Q}$ using the equations $\mathbf{MQ=QJ}$ and (\ref{4}) for each
$\lambda_{i}$, $i=1,\ldots,4$. To this end, let $\mathbf{x}=\left(
\mathbf{x}_{1},\mathbf{x}_{2}\right)  \in\mathbb{R}^{2n}$ with $\mathbf{x}
_{1}=\left(  x_{1},\ldots,x_{n}\right)  $ and $\mathbf{x}_{2}=\left(
x_{n+1},\ldots,x_{2n}\right)  $.

For $\lambda_{1}=1-\alpha$, we have
\[
\mathbf{M}-\lambda_{1}\mathbf{I}_{2n}=\left(
\begin{array}
[c]{cc}
\mathbf{0} & \alpha\mathbf{1}_{n}^{T}\mathbf{a}\\
-\beta\mathbf{1}_{n}^{T}\mathbf{b} & \left(  \alpha-\beta\right)\mathbf{I}_{n}
\end{array}
\right)
\]
and that $\left(  \mathbf{M}-\lambda_{1}\mathbf{I}_{2n}\right)  \mathbf{x}^{T}=\mathbf{0}$ if and only if
\begin{equation}
\mathbf{ax}_{2}^{T}=0\text{ and }\mathbf{x}_{2}^{T}=\dfrac{\beta
\mathbf{bx}_{1}^{T}}{\alpha-\beta}\mathbf{1}_{n}^{T}. \label{12}%
\end{equation}
Since $\mathbf{a}\in\mathcal{C}_{n}$ and $\mathbf{ax}_{2}^{T}=0$, it is clear
that when $\mathbf{x}_{2}^{T}\neq\mathbf{0}$ some component of $\mathbf{x}%
_{2}$ have to be positive and some negative. However, since the sign of
$\frac{\beta\mathbf{bx}^{T}}{\alpha-\beta}$ is fixed, the
second identity in (\ref{12}) thus forces the components of $\mathbf{x}_{2}$
to have the same sign. Therefore, (\ref{12}) holds if and only if
$\mathbf{x}_{2}^{T}=\mathbf{0}$, and this gives $\mathbf{bx}_{1}^{T}=0$. In
other words,
\begin{equation}
\mathbf{\tilde{W}}_{\lambda_{1}}=\ker\left(  \mathbf{M}-\lambda_{1}
\mathbf{I}_{2n}\right)  =\left\{  \mathbf{x}\in\mathbb{R}^{2n}:\mathbf{bx}
_{1}^{T}=0,\mathbf{x}_{2}^{T}=\mathbf{0}\right\}  \text{,} \label{25}
\end{equation}
where $\ker\left(  \mathbf{A}\right)  $ denotes the kernel space of a square
matrix $\mathbf{A}\in\mathcal{M}_{s}$ as a linear mapping $\boldsymbol{v}
\mapsto\mathbf{A}\boldsymbol{v}$ for a column vector $\boldsymbol{v}\in\mathbb{R}^{s}$. Since
$\dim\left(  \mathbf{\tilde{W}}_{\lambda_{1}}\right)  =n-1$, we see that the
eigenspace corresponding to $\lambda_{1}$ is $\mathbf{\tilde{W}}_{\lambda_{1}
}$. Further, it is easy to verify that $\boldsymbol{\varepsilon}_{i}=\left(
-b_{1}^{-1}b_{i+1},\mathbf{e}_{i,n-1},\mathbf{0}\right)  \in\mathbb{R}^{2n}$
for $1\leq i\leq n-1$ is a basis for $\mathbf{\tilde{W}}_{\lambda_{1}}$.

For $\lambda_{2}=1-\beta$, we have
\[
\mathbf{M}-\lambda_{2}\mathbf{I}_{2n}=\left(
\begin{array}
[c]{cc}%
\left(  \beta-\alpha\right)  \mathbf{I}_{n} & \alpha\mathbf{1}_{n}^{T}\mathbf{a}\\
-\beta\mathbf{1}_{n}^{T}\mathbf{b} & \mathbf{0}
\end{array}
\right)  .
\]
So, $\left(  \mathbf{M}-\lambda_{2}\mathbf{I}_{2n}\right)  \mathbf{x}^{T}=\mathbf{0}$ if and only if
\begin{equation}
\mathbf{bx}_{1}^{T}=0\text{ and }\mathbf{x}_{1}^{T}=\dfrac{\alpha
\mathbf{ax}_{2}^{T}}{\alpha-\beta}\mathbf{1}_{n}^{T}. \label{13}
\end{equation}
By the same reasoning for the case of $\lambda_{1}$, we see that
\begin{equation}
\mathbf{\tilde{W}}_{\lambda_{2}}=\ker\left(  \mathbf{M}-\lambda_{2}
\mathbf{I}_{2n}\right)  =\left\{  \mathbf{x}\in\mathbb{R}^{2n}:\mathbf{ax}
_{2}^{T}=0,\mathbf{x}_{1}^{T}=\mathbf{0}\right\}  \label{26}
\end{equation}
with $\dim\left(  \mathbf{\tilde{W}}_{\lambda_{1}}\right)  =n-1$ is the eigenspace
corresponding to $\lambda_{3}$. Further, it is easy to verify that
$\boldsymbol{\tilde{\varepsilon}}_{i}=\left(  \mathbf{0},-a_{1}^{-1}%
a_{i+1},\mathbf{e}_{i,n-1}\right)  \in\mathbb{R}^{2n}$ for $1\leq i\leq n-1$
is a basis for $\mathbf{\tilde{W}}_{\lambda_{2}}$.

For $\lambda _{3,4}=1-2^{-1}\left( \alpha +\beta \right) \pm 2^{-1}\sqrt{\Delta }$, we have
\begin{equation*}
\mathbf{M}-\lambda _{3,4}\mathbf{I}_{2n}=\left(
\begin{array}{cc}
\tau \mathbf{I}_{n} & \alpha \mathbf{1}_{n}^{T}\mathbf{a} \\
-\beta \mathbf{1}_{n}^{T}\mathbf{b} & \delta \mathbf{I}_{n}
\end{array}
\right) \text{,}
\end{equation*}
where $\tau =\frac{\beta -\alpha }{2}\mp \frac{\sqrt{\Delta }}{2}$ and $
\delta =\frac{\alpha -\beta }{2}\mp \frac{\sqrt{\Delta }}{2}$. Let
\begin{equation*}
\mathbf{T}_{2}=\left(
\begin{array}{cc}
\mathbf{I}_{n} & \mathbf{0} \\
\tau ^{-1}\beta \mathbf{1}_{n}^{T}\mathbf{b} & \mathbf{I}_{n}%
\end{array}
\right)
\end{equation*}
and $\mathbf{\tilde{M}}_{2}=\mathbf{T}_{2}\left( \mathbf{M}-\lambda _{3,4}\mathbf{I}_{2n}\right) $. Then
\begin{equation*}
\mathbf{\tilde{M}}_{2}=\left(
\begin{array}{cc}
\mathbf{I}_{n} & \mathbf{0} \\
\tau ^{-1}\beta \mathbf{1}_{n}^{T}\mathbf{b} & \mathbf{I}_{n}
\end{array}%
\right) \left(
\begin{array}{cc}
\tau \mathbf{I}_{n} & \alpha \mathbf{1}_{n}^{T}\mathbf{a} \\
-\beta \mathbf{1}_{n}^{T}\mathbf{b} & \delta \mathbf{I}_{n}
\end{array}
\right) =\left(
\begin{array}{cc}
\tau \mathbf{I}_{n} & \alpha \mathbf{1}_{n}^{T}\mathbf{a} \\
\mathbf{0} & \tau ^{-1}\beta \alpha \mathbf{1}_{n}^{T}\mathbf{a}+\delta\mathbf{I}_{n}
\end{array}%
\right) \text{,}
\end{equation*}
and $\ker \left( \mathbf{M}-\lambda _{3,4}\mathbf{I}_{2n}\right) =\ker
\left( \mathbf{\tilde{M}}_{2}\right) $ since $\left\vert \mathbf{T}_{2}\right\vert =1$. Obviously, $\mathbf{\tilde{M}}_{2}\mathbf{x}^{T}=\mathbf{0}$ if and only if
\begin{equation}
\tau \mathbf{x}_{1}^{T}+\alpha \mathbf{1}_{n}^{T}\mathbf{ax}_{2}^{T}=\mathbf{
0}\text{\ \ \ and \ \ }\left( \tau ^{-1}\delta ^{-1}\beta \alpha \mathbf{1}_{n}^{T}
\mathbf{a}+\mathbf{I}_{n}\right) \mathbf{x}_{2}^{T}=\mathbf{0}\text{.}
\label{19}
\end{equation}
Since
\begin{equation*}
\tau \delta =-4^{-1}\left( \beta -\alpha \mp \sqrt{\Delta }\right) \left(
\beta -\alpha \pm \sqrt{\Delta }\right) =-4^{-1}\left[ \left( \beta -\alpha
\right) ^{2}-\Delta \right] =-\alpha \beta
\end{equation*}
and $\tau ^{-1}\delta ^{-1}\beta \alpha =-1$, the second identity in \eqref{19} becomes
\begin{equation}
\left( \mathbf{1}_{n}^{T}\mathbf{a}-\mathbf{I}_{n}\right) \mathbf{x}_{2}^{T}=%
\mathbf{0}  \label{31}
\end{equation}
Since the matrix $\mathbf{a1}_{n}^{T}=1\mathbf{\ }$has the only eigenvalue $1
$ whose corresponding eigenvector is $1$, the general solution to (\ref
{31}) is $\mathbf{x}_{2}^{T}=c\mathbf{1}_{n}^{T}$ for some $c\in \mathbb{R}$.
Let $\mathbf{R}=\mathbf{I}_{n-1}-\mathbf{1}_{n-1}^{T}\mathbf{a}_{\left( -n\right)
}$, where $\mathbf{a}_{\left( -n\right) }=\left( a_{1},\ldots
,a_{n-1}\right) $. Then $\left\vert \mathbf{R}\right\vert
=1-\sum_{i=1}^{n-1}a_{i}>0$ since $\mathbf{a}\in \mathcal{C}_{n}$ and $%
\mathsf{rank}\left( \mathbf{R}\right) =n-1$. So, $\ker \left( \mathbf{1}%
_{n}^{T}\mathbf{a}-\mathbf{I}_{n}\right) =\mathsf{span}\left( \left\{
\mathbf{1}_{n}\right\} \right) $,
and the general solution to $\left( \mathbf{M}-\lambda _{3,4}\mathbf{I}_{2n}\right) \mathbf{x}^{T}=\mathbf{0}$ is
\begin{equation}
\mathbf{x}_{1}=-\alpha \tau ^{-1}\mathbf{x}_{2}\text{ and }\mathbf{x}_{2}=c\mathbf{1}_{n}\text{.}  \label{20}
\end{equation}
Therefore, the solution space to $\left( \mathbf{M}-\lambda _{3}\mathbf{I}
_{2n}\right) \mathbf{x}^{T}=\mathbf{0}$, i.e., that when $\tau =\frac{\beta
-\alpha -\sqrt{\Delta }}{2}$ and $\delta =\frac{\alpha -\beta -\sqrt{\Delta }}{2}$,
is
\begin{equation}
\mathbf{\tilde{W}}_{\lambda _{3}}=\mathsf{span}\left( \left\{ \mathbf{x}\in
\mathbb{R}^{2n}:\mathbf{x}=\left( \mathbf{1}_{n},-2\alpha \left( \beta
-\alpha -\sqrt{\Delta }\right)^{-1} \mathbf{1}_{n}\right) \right\} \right) \text{.}  \label{21}
\end{equation}
Further, the solution space to $\left( \mathbf{M}-\lambda _{4}\mathbf{I}_{2n}\right) \mathbf{x}^{T}=\mathbf{0}$, i.e., that when $\tau =\frac{\beta-\alpha +\sqrt{\Delta }}{2}$ and $\delta =\frac{\alpha -\beta +\sqrt{\Delta }}{2}$, is
\begin{equation}
\mathbf{\tilde{W}}_{\lambda _{4}}=\mathsf{span}\left( \left\{ \mathbf{x}\in
\mathbb{R}^{2n}:\mathbf{x}=\left( \mathbf{1}_{n},-2\alpha \left( \beta
-\alpha +\sqrt{\Delta }\right)^{-1} \mathbf{1}_{n}\right) \right\} \right) \text{.}  \label{32}
\end{equation}

Combining the solutions to equations $\left(  \mathbf{M}-\lambda_{i}%
\mathbf{I}_{2n}\right)  \mathbf{x}^{T}=\mathbf{0}$ for $1\leq i\leq4$, we see
that (\ref{14}) holds with (\ref{15}). This completes the proof.
\end{proof}

\subsection{Inverse of The Matrix for Basis}

Next we derive the inverse $\mathbf{Q}^{-1}$ of $\mathbf{Q}$ so that the full, explicit
decomposition of $\mathbf{M}$ will be available. Even though it is difficult in general to find explicitly
 the inverse of a large-dimensional matrix, i.e., when $n$ is large or equivalently there are many growth rates involved in model \eqref{3}, the inverse $\mathbf{Q}^{-1}$ terms out to be very simple (see \autoref{Thm:BasisExplicit}) due to the fact that the weights $\mathbf{a}$ and $\mathbf{b}$ both represent convex combinations and lie in the simplex $\mathcal{C}_{n}$.

In order to state the result, we introduce some notations. Let $\tau _{-}=2\left( \beta -\alpha -\sqrt{
\Delta }\right)^{-1} $, $\tau _{+}=2\left( \beta -\alpha +\sqrt{\Delta }
\right)^{-1} $, $\mathbf{a}_{\left( -1\right) }=\left( a_{2},\ldots ,a_{n}\right)
$ and $\mathbf{b}_{\left( -1\right) }=\left( b_{2},\ldots ,b_{n}\right) $.
Recall (\ref{14}) and (\ref{15}).
Then $\mathbf{Q}$ can be written into a $4 \times 4$ block matrix as
\begin{equation}
\mathbf{Q}=\left(
\begin{array}{cc}
\begin{array}{cc}
-b_{1}^{-1}\mathbf{b}_{\left( -1\right) } & 1 \\
\mathbf{I}_{n-1} & \mathbf{1}_{n-1}^{T}
\end{array}
&
\begin{array}{cc}
\mathbf{0}_{n\times \left( n-1\right) } & \text{ \ }\mathbf{1}_{n}^{T}
\end{array}
\\
\begin{array}{cc}
\mathbf{0}_{n\times \left( n-1\right) } & -\alpha \tau _{-}\mathbf{1}_{n}^{T}
\end{array}
&
\begin{array}{cc}
-a_{1}^{-1}\mathbf{a}_{\left( -1\right) } & -\alpha \tau _{+} \\
\mathbf{I}_{n-1} & -\alpha \tau _{+}\mathbf{1}_{n-1}^{T}
\end{array}
\end{array}
\right) ,  \label{24}
\end{equation}
where $\mathbf{0}_{s\times s^{\prime}}\in\mathcal{M}_{s\times s^{\prime}}$ has
all entries as zero. Further, for integers $i$ and $j$ such that $1\leq i\leq
j\leq2n$, let $\mathbf{E}_{ij}\in\mathcal{M}_{2n}$ be such that its $ij$th
entry is $1$ and other entries are identically zero, and for $\hat{c}
\in\mathbb{R}$ let $\mathbf{P}_{i,j}\left(  \hat{c}\right)  =\mathbf{I}_{2n}+\hat{c}\mathbf{E}_{i,j}$.
Note that for any matrix $\mathbf{A}=\left(  \tilde{a}_{ij}\right)
\in\mathcal{M}_{2n}$ the $j$th column of $\mathbf{AE}_{ij}$ is the $i$th
column of $\mathbf{A}$ and all other entries of $\mathbf{AE}_{ij}$ are zero.

\begin{theorem}
\label{Thm:BasisExplicit}Under the conditions of \autoref{Thm:Basis}, the
inverse $\mathbf{Q}^{-1}$ of $\mathbf{Q}$ is
\begin{equation}
\mathbf{Q}^{-1}=\mathbf{P}_{n,2n}\left( -1\right) \mathbf{P}_{2n,n}\mathbf{P}
_{2n,n}\left( \tilde{\tau}^{-1}\alpha \tau _{-}\right) \mathsf{diag}\left\{
\mathbf{Q}_{21}^{-1},\mathbf{Q}_{22}^{-1}\right\} \text{,}  \label{39}
\end{equation}
where $\tilde{\tau}=\alpha \left( \tau _{-}-\tau _{+}\right) $,
\begin{equation}
\mathbf{Q}_{21}^{-1}=\left(
\begin{array}{cc}
-b_{1}\mathbf{1}_{n-1}^{T} & \mathbf{I}_{n-1}-\mathbf{1}_{n-1}^{T}\mathbf{b}
_{\left( -1\right) } \\
b_{1} & \mathbf{b}_{\left( -1\right) }
\end{array}
\right)   \label{33}
\end{equation}
and
\begin{equation}
\mathbf{Q}_{22}^{-1}=\left(
\begin{array}{cc}
-a_{1}\mathbf{1}_{n-1}^{T} & \mathbf{I}_{n-1}-\mathbf{1}_{n-1}^{T}\mathbf{a}
_{\left( -1\right) } \\
\tilde{\tau}^{-1}a_{1} & \tilde{\tau}^{-1}\mathbf{a}_{\left( -1\right) }
\end{array}
\right) \text{.}  \label{35}
\end{equation}
\end{theorem}

\begin{proof}
Multiplying the $n$-th column of $\mathbf{Q}$ by $-1$ and adding the resultant
column to the $2n$-th column of $\mathbf{Q}$ gives
\begin{equation*}
\mathbf{Q}_{1}=\mathbf{QP}_{n,2n}\left( -1\right) =\left(
\begin{array}{cc}
\begin{array}{cc}
-b_{1}^{-1}\mathbf{b}_{\left( -1\right) } & 1 \\
\mathbf{I}_{n-1} & \mathbf{1}_{n-1}^{T}
\end{array}
& \mathbf{0}_{n\times n} \\
\begin{array}{cc}
\mathbf{0}_{n\times \left( n-1\right) } & -\alpha \tau _{-}\mathbf{1}_{n}^{T}
\end{array}
&
\begin{array}{cc}
-a_{1}^{-1}\mathbf{a}_{\left( -1\right) } & \tilde{\tau} \\
\mathbf{I}_{n-1} & \tilde{\tau}\mathbf{1}_{n-1}^{T}
\end{array}
\end{array}
\right)
\end{equation*}
with $\tilde{\tau}=\alpha \left( \tau _{-}-\tau _{+}\right) $.
Multiplying the $2n$-th column of $\mathbf{Q}_{1}$ by $\tilde{\tau}^{-1}\alpha
\tau _{-}$ and adding the resultant column to the $n$-th column of $\mathbf{Q}_{1}$ gives
\begin{equation*}
\mathbf{Q}_{2}=\mathbf{Q}_{1}\mathbf{P}_{2n,n}\left( -\Delta ^{-1/2}\alpha
\tau _{-}\right) =\left(
\begin{array}{cc}
\begin{array}{cc}
-b_{1}^{-1}\mathbf{b}_{\left( -1\right) } & 1 \\
\mathbf{I}_{n-1} & \mathbf{1}_{n-1}^{T}
\end{array}
& \mathbf{0}_{n\times n} \\
\mathbf{0}_{n\times n} &
\begin{array}{cc}
-a_{1}^{-1}\mathbf{a}_{\left( -1\right) } & \tilde{\tau} \\
\mathbf{I}_{n-1} & \tilde{\tau}\mathbf{1}_{n-1}^{T}
\end{array}
\end{array}
\right) \text{.}
\end{equation*}
Let%
\begin{equation}
\mathbf{Q}_{21}=\left(
\begin{array}{cc}
-b_{1}^{-1}\mathbf{b}_{\left( -1\right) } & 1 \\
\mathbf{I}_{n-1} & \mathbf{1}_{n-1}^{T}
\end{array}
\right) \text{ and }\mathbf{Q}_{22}=\left(
\begin{array}{cc}
-a_{1}^{-1}\mathbf{a}_{\left( -1\right) } & \tilde{\tau} \\
\mathbf{I}_{n-1} & \tilde{\tau}\mathbf{1}_{n-1}^{T}
\end{array}
\right) \text{.}  \label{36}
\end{equation}
Then $\mathbf{Q}_{2}=\mathsf{diag}\left\{ \mathbf{Q}_{21},\mathbf{Q}_{22}\right\} $ and
\begin{equation}
\mathbf{Q}^{-1}=\mathbf{P}_{n,2n}\left( -1\right) \mathbf{P}_{2n,n}\mathbf{P}
_{2n,n}\left( \tilde{\tau}^{-1}\alpha \tau _{-}\right) \mathbf{Q}_{2}^{-1}\text{.}  \label{37}
\end{equation}
Therefore, it suffices to find $\mathbf{Q}_{2}^{-1}=\mathsf{diag}\left\{
\mathbf{Q}_{21}^{-1},\mathbf{Q}_{22}^{-1}\right\} $ or equivalently to find
$\mathbf{Q}_{21}^{-1}$ and $\mathbf{Q}_{22}^{-1}$.

Let
\begin{equation*}
\mathbf{R}_{1}=\left(
\begin{array}{cc}
\mathbf{0}_{\left( n-1\right) \times 1} & \mathbf{I}_{n-1} \\
1 & \mathbf{0}_{1\times \left( n-1\right) }
\end{array}%
\right) \text{.}
\end{equation*}%
Then $\mathbf{R}_{1}^{-1}=\mathbf{R}_{1}^{T}$, left multiplication by $%
\mathbf{R}_{1}$ permutes the rows, and right multiplication by $\mathbf{R}_{1}$
permutes the columns. Further,
\begin{equation*}
\mathbf{\tilde{Q}}_{21}=\mathbf{R}_{1}\mathbf{Q}_{21}=\left(
\begin{array}{cc}
\mathbf{I}_{n-1} & \mathbf{1}_{n-1}^{T} \\
-b_{1}^{-1}\mathbf{b}_{\left( -1\right) } & 1
\end{array}
\right) \text{.}
\end{equation*}
Since $1+b_{1}^{-1}\mathbf{b}_{\left( -1\right) }\mathbf{1}_{n-1}^{T}=b_{1}^{-1}\neq 0$, we have
\begin{equation}
\mathbf{\tilde{Q}}_{21}^{-1}=\left(
\begin{array}{cc}
\mathbf{I}_{n-1}-\mathbf{1}_{n-1}^{T}\mathbf{b}_{\left( -1\right) } & -b_{1}\mathbf{1}_{n-1}^{T} \\
\mathbf{b}_{\left( -1\right) } & b_{1}
\end{array}
\right)   \label{38}
\end{equation}
and $\mathbf{Q}_{21}^{-1}=\mathbf{\tilde{Q}}_{21}^{-1}\mathbf{R}_{1}$, which implies (\ref{33}). To get $\mathbf{Q}_{22}^{-1}$, we start from
\begin{equation*}
\mathbf{\tilde{Q}}_{22}=\mathbf{R}_{1}\mathbf{Q}_{22}=\left(
\begin{array}{cc}
\mathbf{I}_{n-1} & \tilde{\tau}\mathbf{1}_{n-1}^{T} \\
-a_{1}^{-1}\mathbf{a}_{\left( -1\right) } & \tilde{\tau}
\end{array}
\right) \text{.}
\end{equation*}
Since $\tilde{\tau}+a_{1}^{-1}\mathbf{a}_{\left( -1\right) }\tilde{\tau}
\mathbf{1}_{n-1}^{T}=\tilde{\tau}a_{1}^{-1}\neq 0$, we see
\begin{equation}
\mathbf{\tilde{Q}}_{22}^{-1}=\left(
\begin{array}{cc}
\mathbf{I}_{n-1}-\mathbf{1}_{n-1}^{T}\mathbf{a}_{\left( -1\right) } & -a_{1}%
\mathbf{1}_{n-1}^{T} \\
\tilde{\tau}^{-1}\mathbf{a}_{\left( -1\right) } & \tilde{\tau}^{-1}a_{1}%
\end{array}%
\right)   \label{34}
\end{equation}%
and $\mathbf{Q}_{22}^{-1}=\mathbf{\tilde{Q}}_{22}^{-1}\mathbf{R}_{1}$, which implies (\ref{35}).

Combining (\ref{33}), (\ref{35}) and (\ref{37}), we get (\ref{39}), which completes the proof.
\end{proof}

\autoref{Thm:BasisExplicit} shows that $\mathbf{Q}^{-1}$ has an easy and explicit form that allows quick
computation even when $n$ is large, since $\mathbf{Q}_{21}^{-1}$ and $\mathbf{Q}_{22}^{-1}$
are very simple and $\mathbf{P}_{n,2n}\left( -1\right) $
and $\mathbf{P}_{2n,n}\left( \tilde{\tau}^{-1}\alpha \tau
_{-}\right) $ are only two linear operations on two columns of $\mathsf{diag}
\left\{ \mathbf{Q}_{21}^{-1},\mathbf{Q}_{22}^{-1}\right\} $. The inverse $
\mathbf{Q}^{-1}$ helps give the explicit decomposition of $\mathbf{M}$ and the explicit solution $\left\{ \boldsymbol{z}_{t}\right\} _{t\in \mathbb{Z}_{+}}$
in \autoref{Cor:SolutionExplicit} that reveals its long term behavior.

\section{The Explicit Solution and Its Properties\label{Sec:Solution}}

We are ready to provide the explicit solution $\left\{ \boldsymbol{z}_{t}\right\} _{t\in \mathbb{Z}_{+}}$ to model (\ref{3}) using the explicit decomposition of $\mathbf{M}$ in terms of $\mathbf{J}$, $\mathbf{Q}$ and $\mathbf{Q}^{-1}$ given in
\autoref{Thm:JordanForm}, \autoref{Thm:Basis}, and \autoref{Thm:BasisExplicit}.

\begin{corollary}
\label{Cor:SolutionExplicit}Under the conditions of \autoref{Thm:Basis},
model (\ref{3}) has the explicit solution
\begin{equation}
\boldsymbol{z}_{t+1}=\mathbf{QJ}^{t+1}\mathbf{Q}^{-1}\boldsymbol{z}_{0}
+\sum_{i=0}^{t}\mathbf{QJ}^{i}\mathbf{Q}^{-1}\boldsymbol{\gamma}_{t-i}
\label{28}
\end{equation}
and explicit, equivalent solution
\begin{equation}
\boldsymbol{\tilde{z}}_{t+1}=\mathbf{J}^{t+1}\boldsymbol{\tilde{z}}_{0}
+\sum_{i=0}^{t}\mathbf{J}^{i}\boldsymbol{\tilde{\gamma}}_{t-i}\text{,}
\label{5}
\end{equation}
where $\boldsymbol{\tilde{z}}_{t}=\mathbf{Q}^{-1}\boldsymbol{z}_{t}$ and
$\boldsymbol{\tilde{\gamma}}_{t}=\mathbf{Q}^{-1}\boldsymbol{\gamma}_{t}$ for
$t\in\mathbb{Z}_{+}$.
\end{corollary}

\begin{proof}
By results in \autoref{Sec:CanonicalForm}, model (\ref{3}) is just%
\[
\boldsymbol{z}_{t+1}=\mathbf{M}^{t+1}\boldsymbol{z}_{0}+\sum_{i=0}%
^{t}\mathbf{M}^{i}\boldsymbol{\gamma}_{t-i}%
\]
with the initial value $\boldsymbol{z}_{0}$, where $\mathbf{Q}$ is given in
(\ref{14}), $\mathbf{J}$ in (\ref{4}), and $\mathbf{Q}^{-1}$ in (\ref{39}).
This implies (\ref{28}) and (\ref{5}), and completes the proof.
\end{proof}

In other words, $\left\{  \boldsymbol{z}_{t}\right\}  _{t\in\mathbb{Z}}$ can
be represented almost as a vector moving average model of order $t-1$ with
independent and identically distributed (i.i.d.) errors $\left\{
\boldsymbol{\gamma}_{t}\right\}  _{t\in\mathbb{Z}}$.

\subsection{Nonstationarity}

Recall that a stochastic process is second-order stationary if its
covariance function of a fixed lag depends only on the lag but not on the
time index. In order to study the behavior of $\left\{ \boldsymbol{z}
_{t}\right\} _{t\in \mathbb{Z}}$, we need the following lemma on equivalence
of second-order stationarity.

\begin{lemma}
\label{Lm:EquivalenceOnStationarity}Under the conditions of
\autoref{Thm:Basis}, both or neither of the sequence $\left\{
\boldsymbol{\tilde{z}}_{t}\right\}  _{t\in\mathbb{Z}_{+}}$ defined in
\autoref{Cor:SolutionExplicit} and $\left\{  \boldsymbol{z}_{t}\right\}
_{t\in\mathbb{Z}_{+}}$ are second-order stationary.
\end{lemma}

\begin{proof}
Let
\begin{equation}
\mathbf{\Gamma}_{t+\tau',t}=\mathbb{E}\left[  \left(  \boldsymbol{z}_{t+\tau'
}-\mathbb{E}\left(  \boldsymbol{z}_{t+\tau'}\right)  \right)  \left(
\boldsymbol{z}_{t}-\mathbb{E}\left(  \boldsymbol{z}_{t}\right)  \right)
^{T}\right]  \label{29}
\end{equation}
for $t,\tau'\in\mathbb{Z}_{+}$, where $\mathbb{E}$ the expectation with respect
to $\mathbb{P}$. Then $\mathbb{E}\left(  \boldsymbol{\tilde{z}}_{t}\right)
=\mathbf{Q}^{-1}\mathbb{E}\left(  \boldsymbol{z}_{t}\right)  $ and
\begin{equation}
\mathbf{\tilde{\Gamma}}_{t+\tau',t}=\mathbb{E}\left[  \left(
\boldsymbol{\tilde{z}}_{t+\tau'}-\mathbb{E}\left(  \boldsymbol{\tilde{z}}
_{t}\right)  \right)  \left(  \boldsymbol{\tilde{z}}_{t+\tau'}-\mathbb{E}
\left(  \boldsymbol{\tilde{z}}_{t}\right)  \right)  ^{T}\right]
=\mathbf{Q}^{-1}\mathbf{\Gamma}_{t+\tau',t}\left(  \mathbf{Q}^{-1}\right)
^{T}\text{.}\label{40}
\end{equation}
This, together with the nonsingularity of $\mathbf{Q}$, implies that either both
or neither $\left\{  \boldsymbol{\tilde{z}}_{t}\right\}  _{t\in\mathbb{Z}_{+}}$
and $\left\{  \boldsymbol{z}_{t}\right\}  _{t\in\mathbb{Z}_{+}}$ are second-order
stationary. This completes the proof.
\end{proof}

By \autoref{Lm:EquivalenceOnStationarity}, it suffices to study the second-order
stationarity of $\left\{  \boldsymbol{\tilde{z}}_{t}\right\}  _{t\in
\mathbb{Z}}$. By the assumptions on $\left\{  \varepsilon_{i}\rndbr{t}\right\}
_{i=1}^{2n}$ given in \autoref{Sec:Intro}, we have the mean vector of
$\boldsymbol{\gamma}_{t}$ as%
\[
\boldsymbol{\mu}_{t}=\left(  \alpha\mu_{1},\ldots,\alpha\mu_{n},-\beta\mu
_{n+1},\ldots,-\beta\mu_{2n}\right)  ^{T}
\]
and the covariance matrix $\mathbf{\Sigma}_{t}$ of $\boldsymbol{\gamma}_{t}$ as
\[
\mathbf{\Sigma}_{t}=\mathsf{diag}\left\{  \alpha^{2}\sigma_{1}^{2}
,\ldots,\alpha^{2}\sigma_{n}^{2},\beta^{2}\sigma_{n+1}^{2},\ldots,\beta
^{2}\sigma_{2n}^{2}\right\}  \text{.}%
\]
The following result shows that $\left\{ \boldsymbol{z}_{t}\right\} _{t\in
\mathbb{Z}_{+}}$ is not second-order stationary.

\begin{proposition}
\label{Prob:nonStat}Suppose $\boldsymbol{z}_{0}$ is independent of the
sequence $\left\{ \boldsymbol{\gamma }_{t}\right\} _{t\in \mathbb{Z}_{+}}$
and has covariance matrix $\mathbf{G}$. Then $\mathbf{\tilde{\Gamma}}
_{t+\tau ^{\prime },t}$ in (\ref{40}) for $t \geq 2$ satisfies
\begin{equation}
\mathbf{\tilde{\Gamma}}_{t+\tau ^{\prime },t}=\mathbf{J}^{t+\tau ^{\prime }}
\mathbf{GJ}^{t}+\sum_{i=1}^{t-1}\mathbf{J}^{\tau ^{\prime }+i}\mathbf{\Sigma
}_{0}\mathbf{J}^{i}\text{.}  \label{40a}
\end{equation}
Therefore, neither $\left\{ \boldsymbol{z}_{t}\right\} _{t\in \mathbb{Z}}$
nor $\left\{ \boldsymbol{\tilde{z}}_{t}\right\} _{t\in \mathbb{Z}_{+}}$ is
second-order stationary.
\end{proposition}

\begin{proof}
To compute $\mathbf{\tilde{\Gamma}}_{t+\tau ^{\prime },t}$, it suffices to assume that $\boldsymbol{z}_{0}$ and
each $\boldsymbol{\gamma }_{t},t\in \mathbb{Z}_{+}$ has mean zero but with
their corresponding covariances. Namely, it suffices to assume%
\begin{equation*}
\boldsymbol{\tilde{z}}_{t+1}=\mathbf{J}^{t+1}\boldsymbol{\hat{z}}
_{0}+\sum_{i=0}^{t}\mathbf{J}^{i}\boldsymbol{\hat{\gamma}}_{t-i},
\end{equation*}
where $\boldsymbol{\hat{z}}_{0}$ is the mean centered $\boldsymbol{z}_{0}$
and each $\boldsymbol{\hat{\gamma}}_{t}$ is the mean centered $\boldsymbol{\gamma }_{t}$. This implies that
\begin{equation*}
\mathbf{\tilde{\Gamma}}_{t+\tau ^{\prime },t}=\mathbb{E}\left[ \boldsymbol{
\tilde{z}}_{t+\tau ^{\prime }}\boldsymbol{\tilde{z}}_{t}^{T}\right] =
\mathbb{E}\left[ \left( \mathbf{J}^{t+\tau ^{\prime }}\boldsymbol{\hat{z}}
_{0}+\sum\nolimits_{i=0}^{t+\tau ^{\prime }-1}\mathbf{J}^{i}\boldsymbol{\hat{
\gamma}}_{t-i}\right) \left( \mathbf{J}^{t}\boldsymbol{\hat{z}}
_{0}+\sum\nolimits_{i=0}^{t-1}\mathbf{J}^{i}\boldsymbol{\hat{\gamma}}
_{t-i}\right) \right] ,
\end{equation*}
which simplifies into \eqref{40a}. Since $\mathbf{\tilde{\Gamma}}_{t+\tau ^{\prime
},t}$ depends on $t$, $\left\{ \boldsymbol{\tilde{z}}_{t}\right\} _{t\in
\mathbb{Z}_{+}}$ is not second-order stationary, and by  \autoref
{Lm:EquivalenceOnStationarity} nor is $\left\{ \boldsymbol{z}_{t}\right\}
_{t\in \mathbb{Z}_{+}}$. This completes the proof.
\end{proof}

\subsection{Limiting Behavior}\label{subsec:limit}

For the representations given in \eqref{28} and \eqref{5}, it is easy to explore the
long term behavior of $\left\{ \boldsymbol{\tilde{z}}_{t}\right\} _{t\in
\mathbb{Z}_{+}}$ than that of $\left\{ \boldsymbol{z}_{t}\right\} _{t\in
\mathbb{Z}}$. Recall
\begin{equation*}
\mathbf{J}=\mathsf{diag}\left\{ \lambda _{1}\mathbf{I}_{n-1},\lambda _{3}
\mathbf{I}_{1},\lambda _{2}\mathbf{I}_{n-1},\lambda _{4}\mathbf{I}%
_{1}\right\}
\end{equation*}
for which $\lambda _{1}=1-\alpha $, $\lambda _{2}=1-\beta $, $\lambda
_{3,4}=2^{-1}\left( 2-\alpha -\beta \pm \sqrt{\Delta }\right) $.

\begin{corollary}
\label{Cor:LongTermProperty}Under the conditions of \autoref{Thm:Basis}, if
moreover
\begin{equation}
0<\max_{1\leq i\leq 4}\left\vert \lambda _{i}\right\vert <1,  \label{45}
\end{equation}%
then, as $t\rightarrow \infty $,
\begin{equation}
\boldsymbol{\tilde{z}}_{t+1} \overset{\mathcal{D}}{\rightarrow}
\mathsf{diag}\left\{ \tilde{\lambda}_{1}\mathbf{I}_{n-1},\tilde{\lambda}_{3}
\mathbf{I}_{1},\tilde{\lambda}_{2}\mathbf{I}_{n-1},\tilde{\lambda}_{4}
\mathbf{I}_{1}\right\} \boldsymbol{\tilde{\gamma}}_{0}  \label{43}
\end{equation}
and
\begin{equation*}
\boldsymbol{z}_{t+1} \overset{\mathcal{D}}{\rightarrow}
\mathbf{Q}\mathsf{diag}\left\{ \tilde{\lambda}_{1}\mathbf{I}_{n-1},\tilde{
\lambda}_{3}\mathbf{I}_{1},\tilde{\lambda}_{2}\mathbf{I}_{n-1},\tilde{\lambda
}_{4}\mathbf{I}_{1}\right\} \mathbf{Q}^{-1}\boldsymbol{\gamma }_{0}\text{,}
\end{equation*}
where $\overset{\mathcal{D}}{\rightarrow}
$ denotes convergence in distribution and $\tilde{\lambda}_{i}=\left(
1-\lambda _{i}\right) ^{-1}$ for $1\leq i\leq 4$.
\end{corollary}

The proof of \autoref{Cor:LongTermProperty} is omitted as it follows
immediately from the convergence of $\sum_{i=0}^{\infty }\lambda
_{j}^{i}=\left( 1-\lambda _{j}\right) ^{-1}\neq 0$ for $1\leq j\leq 4$ when (
\ref{45}) holds, the fact that $\boldsymbol{\tilde{z}}_{t}=\mathbf{Q}^{-1}%
\boldsymbol{z}_{t}$ and $\boldsymbol{\tilde{\gamma}}_{t}=\mathbf{Q}^{-1}%
\boldsymbol{\gamma }_{t}$ for $t\in \mathbb{Z}_{+}$, and the i.i.d. property of $\left\{ \boldsymbol{\tilde{\gamma
}}_{t}\right\} _{t\in \mathbb{Z}_{+}}$. It is clear that any $\lambda _{i}>1$
for some $1\leq i\leq 4$ lead to the explosion in the variance of the
corresponding subvector of $\boldsymbol{\tilde{z}}_{t}$ and that of some subvector of $\boldsymbol{z}_{t}$ as $t\rightarrow
\infty $. Further, when $\lambda _{i},1\leq i\leq 4$ have different signs,
oscillations in components of $\boldsymbol{\tilde{z}}_{t}$ and those of $
\boldsymbol{z}_{t}$ will be induced. However, a complete analysis of such oscillations
seems to be difficult to perform for the original series $\left\{ \boldsymbol{z}_{t}\right\} _{t\in
\mathbb{Z}}$ due to the term $\mathbf{J}^{i}\mathbf{Q}^{-1}\boldsymbol{\gamma}_{t-i}$
in \eqref{28}.

\section{The Induced Model for Business Cycles}\label{sec:busiCyc}

Recall that the aim of proposing model (\ref{3}) in \cite{Ormerod:2001} is to induce a model for business
cycles in the economy created by the growth rates $x_i\rndbr{t}$ under the influence of the sentiments $y_i\rndbr{t}$.
Such an induced model has been given in \cite{Ormerod:2001} but with a miscalculated forcing term (see identity (10.9A) therein) and is called a ``damped pendulum''. Despite its being call so, the induced model for business cycles is similar to an autoregressive (AR) model of order $2$ in $\bar{x}_t$ for $t \in \mathbb{Z}_{+}$, which we now provide and analyze.


From model (\ref{3}), we have, by weighting corresponding with $\mathbf{a}$ and $\mathbf{b}$,
\begin{equation}
\left\{
\begin{array}{c}
\bar{x}\left( t+1\right) =\left( 1-\alpha \right) \bar{x}\left( t\right)
+\alpha \bar{y}\left( t\right) +\alpha \bar{\varepsilon}\left( t\right) , \\
\bar{y}\left( t+1\right) =\left( 1-\beta \right) \bar{y}\left( t\right)
-\beta \bar{x}\left( t\right) -\beta \bar{\eta}\left( t\right) ,
\end{array}
\right.   \label{46}
\end{equation}%
where $\bar{\varepsilon}\left( t\right) =\sum_{i=1}^{n}b_{i}\varepsilon
_{i}\left( t\right) $ and $\bar{\eta}\left( t\right)
=\sum_{i=1}^{n}a_{i}\eta _{i}\left( t\right) $.
From the first identity in \eqref{46}, we obtain
\begin{equation}
\bar{y}\left( t\right) =\alpha ^{-1}\left[ \bar{x}\left( t+1\right) -\left(
1-\alpha \right) \bar{x}\left( t\right) -\alpha \bar{\varepsilon}\left(
t\right) \right] =\alpha ^{-1}\Delta \bar{x}\left( t\right) +\bar{x}\left(
t\right) -\bar{\varepsilon}\left( t\right) ,  \label{47}
\end{equation}%
where $\Delta \bar{x}\left( t\right) =\bar{x}\left( t+1\right) -\bar{x}
\left( t\right) $. Plugging \eqref{47} back into the second identity of \eqref{46}, we get
\begin{eqnarray}
&&\alpha ^{-1}\Delta \bar{x}\left( t+1\right) -\alpha ^{-1}\Delta \bar{x}
\left( t\right) +\bar{x}\left( t+1\right) -\bar{x}\left( t\right) +\beta
\alpha ^{-1}\Delta \bar{x}\left( t\right) +2\beta \bar{x}\left( t\right)
\label{48} \\
&=&\bar{\varepsilon}\left( t+1\right) -\left( 1-\beta \right) \bar{
\varepsilon}\left( t\right) -\beta \bar{\eta}\left( t\right) .  \notag
\end{eqnarray}
After simplification, \eqref{48} becomes
\begin{equation}
\Delta ^{2}\bar{x}\left( t\right) +\left( \alpha +\beta \right) \Delta \bar{x
}\left( t\right) +2\alpha \beta \bar{x}\left( t\right) =h\left( t\right) ,
\label{49}
\end{equation}%
where $\Delta ^{2}\bar{x}\left( t\right) =\Delta \bar{x}\left( t+1\right)
-\Delta \bar{x}\left( t\right) $, $\Delta \bar{\varepsilon}\left( t\right) =
\bar{\varepsilon}\left( t+1\right) -\bar{\varepsilon}\left( t\right) $ and
\begin{equation}
h\left( t\right) =\alpha \Delta \bar{\varepsilon}\left( t\right) +\alpha
\beta \left[ \bar{\varepsilon}\left( t\right) -\bar{\eta}\left( t\right)\right] .  \label{50}
\end{equation}
Equations (\ref{49}) and (\ref{50}) together describe what is called in \cite{Ormerod:2001} a ``damped pendulum'', for which
$h\left( t\right)$ is the forcing term. Note however that $h\left( t\right)$ in \eqref{50}, the correct one, is different
than the mistaken one in identity (10.9A) therein.

\subsection{The Periodic Solution}

On the other hand, (\ref{49}) and (\ref{50}) almost form a second order difference equation in $\bar{x}_t$ for $t \in \mathbb{Z}_{+}$ except that the random error $h\left(t\right)$ involves a term at time $t+1$. Specifically,
\begin{equation}
\bar{x}\left( t+2\right) +\left( \alpha +\beta -2\right) \bar{x}\left(
t+1\right) +\left( 1-\alpha -\beta + 2 \alpha \beta\right) \bar{x}\left( t\right) =h\left(
t\right) \text{.}  \label{51}
\end{equation}
It should be noted that \eqref{51} is not an autoregressive model of order $2$ since $h\rndbr{t}$ involves $\bar{\varepsilon}\rndbr{t+1}$ at time $t+1$.
To explore the properties of $\clbr{\bar{x}_{t}}_{t \in \mathbb{Z}_{+}}$, let $\kappa _{1}=\alpha +\beta -2$, $\kappa _{2}=1-\alpha -\beta+2 \alpha \beta$, and the homogeneous version of \eqref{51} be
\begin{equation}
\bar{x}\left( t+2\right) +\kappa _{1}\bar{x}\left( t+1\right) +\kappa _{2}
\bar{x}\left( t\right) =0.  \label{52}
\end{equation}
Then the characteristic polynomial for both \eqref{51} and \eqref{52} is
\begin{equation*}
q\left( w\right) =w^{2} + \kappa_1 w+ \kappa_2,
\end{equation*}
which has roots $\rho _{1,2}=\frac{-\kappa _{1}\pm \sqrt{\Delta_{1}}}{2}$
with $\Delta _{1}=\kappa _{1}^{2}-4\kappa _{2}=\Delta =\alpha ^{2}+\beta^{2}-6\alpha \beta $ (note that $\Delta$ is defined right before \autoref{Thm:JordanForm}).

Let $\mathcal{L}$ be the lag operator of order one, and recall $d_{1}=\left( 3-2\sqrt{2}\right)\beta$
and $d_{2}=\left( 3+2\sqrt{2}\right) \beta$. We have the following result that
describes when $\clbr{\bar{x}_{t}}_{t \in \mathbb{Z}_{+}}$ can display
periodic behavior and gives the solution $\clbr{\bar{x}_{t}}_{t \in \mathbb{Z}_{+}}$.

\begin{theorem}\label{thm:Pediodic}
Set $\omega =-\arctan \left(\kappa_1^{-1} \sqrt{\vert\Delta _{1}\vert}\right)$. For model (\ref{52}), if
\begin{equation}\label{56a}
  \min\clbr{d_1,d_2} < \alpha < \max\clbr{d_1,d_2},
\end{equation}
then
the general, periodic solution $\clbr{\bar{x}_{t}}_{t \in \mathbb{Z}_{+}}$ is
\begin{equation}
\bar{x}\left( t\right) =c_{1}\left\vert \rho _{1}\right\vert ^{t}\cos \left(c_{2}+\omega t\right) \label{53a}
\end{equation}
for some constants $c_{1}$ and $c_{2}$. If additionally
\begin{equation}
\frac{\beta -1}{2\beta -1} <\alpha < \frac{\beta }{2\beta -1} \text{ and }\beta>\frac{1}{2}\label{56}
\end{equation}
or
\begin{equation}
\frac{\beta }{2\beta -1} <\alpha < \frac{\beta -1}{2\beta -1} \text{ and } \beta<\frac{1}{2}, \label{57}
\end{equation}
holds, then the general, periodic solution $\clbr{\bar{x}_{t}}_{t \in \mathbb{Z}_{+}}$ is
\begin{equation}
\bar{x}\left( t\right) =c_{1}\left\vert \rho _{1}\right\vert ^{t}\cos \left(
c_{2}+\omega t\right) +\left( 1-\rho _{1}\mathcal{L}\right) ^{-1}\left(
1-\rho _{2}\mathcal{L}\right) ^{-1}h\left( t\right), \label{53b}
\end{equation}
where the constants $c_{1}$ and $c_{2}$ can be determined from the initial
values $\bar{x}\left( 0\right) $ and $\bar{x}\left( 1\right) $.
\end{theorem}

\begin{proof}

WLOG, assume $d_1 = \min\clbr{d_1,d_2}$ and $d_2 = \max\clbr{d_1,d_2}$. There are three cases for the general solution to \eqref{52}:

\begin{enumerate}
\item $\Delta _{1}=0$ if and only if $\alpha =d_{1}$ or $\alpha =d_{2}$. In this
case, $\rho _{1}=\rho _{2}=-2^{-1}\kappa _{1}$ and
$\bar{x}\left( t\right) =\left( c_{0}+c_{1}t\right) \rho _{1}^{t}$
is the general solution to (\ref{52}) for some constants $c_0$ and $c_1$.

\item $\Delta _{1}>0$ if and only if $\alpha <d_{1}$ or $\alpha >d_{2}$.
  In this case, $\bar{x}\left( t\right) =c_{1}\rho _{1}^{t}+$ $c_{2}\rho_{2}^{t}$
  is the general solution to (\ref{52}) for some constants $c_1$ and $c_2$.

\item $\Delta _{1}<0$ if and only if $d_{1}<\alpha <d_{2}$. In this case, let
$\rho _{1}=\left\vert \rho _{1}\right\vert e^{\mathrm{i}\omega }$ with $
\omega =-\arctan \left(\kappa_1^{-1} \sqrt{\vert\Delta _{1}\vert}\right) \in (-\pi
,\pi ]$, where $\mathrm{i}^2=-1$. Then $\rho _{1}=\left\vert \rho _{1}\right\vert e^{-\mathrm{i}\omega }$ and
\eqref{53a} is the general, periodic solution to (\ref{52}) for some constants $c_1$ and $c_2$.
\end{enumerate}

So, it is left to find a special solution to (\ref{51}) to obtain \eqref{53b}.
When $\left\vert \rho _{1}\right\vert
<1$, the operators $\left( 1-\rho _{j}\mathcal{L}\right) ,j=1,2$ are
invertible and the inverses $\left(1-\rho _{j}\mathcal{L}\right)^{-1}=\sum_{s=0}^{\infty }\rho _{j}^{s}\mathcal{L}^{s}$ for $j=1,2$, where $\mathcal{L}^{s}$ is the composition of $\mathcal{L}$ by itself $s$ times. However,
$\left\vert \rho _{1}\right\vert^2=\kappa_2$, and  $0 < \left\vert \rho _{1}\right\vert<1$ if and only if
$0 < 1-\alpha-\beta +2\alpha \beta <1$, which holds when $\frac{\beta -1}{2\beta -1}%
<\alpha <\frac{\beta }{2\beta -1}$ and $\beta >\frac{1}{2}$ or when $\frac{
\beta }{2\beta -1}<\alpha <\frac{\beta -1}{2\beta -1}$ and $\beta <\frac{1}{2}$.
Therefore, when additionally \eqref{56} or \eqref{57} holds, we have \eqref{53b}.
This completes the proof.
\end{proof}

Note that pairs $\rndbr{\alpha,\beta}$ satisfying \eqref{56a} and \eqref{56}, or \eqref{56a} and \eqref{57} in \autoref{thm:Pediodic} do exist, which means that the solution \eqref{53b} always exits.
A trajectory from the model \eqref{51}  is displayed in \autoref{FigA}:
\begin{figure}[b]
\centering
\includegraphics[width=0.95\textwidth,height=0.35\textheight]{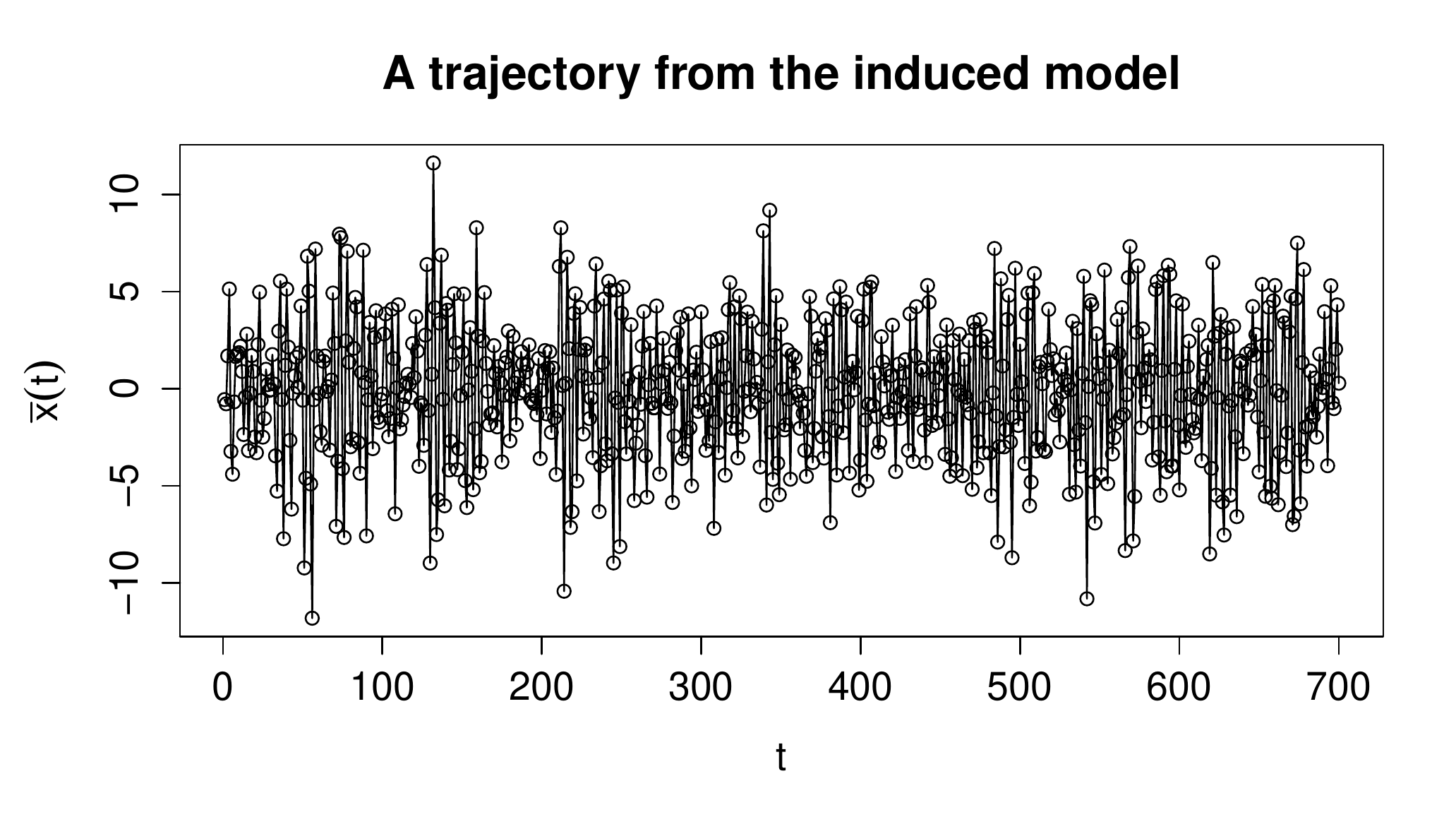}
\caption[A trajectory]{A trajectory of $\bar{x}\rndbr{t}$ for $t=0,\ldots,700$ simulated from the model \eqref{51} with $\alpha = 1.09804$ and $\beta = 0.7$, where $\bar{\varepsilon}\rndbr{t} \propto\mathsf{N}\rndbr{0,1}$ and
$\bar{\eta}\rndbr{t} \propto\mathsf{N}\rndbr{0,1.6^2}$ for each $t$.
The trajectory shows clearly that the periodicity of $\bar{x}\rndbr{t}$ in \eqref{53b} is subject to random perturbation.}
\label{FigA}
\end{figure}


\section{Discussion\label{Sec:Discussion}}

For model (\ref{3}), we have provided the explicit decomposition of its
transition matrix $\mathbf{M}$, the explicit solution, and two key properties of this solution.
In addition, we have provided and analyzed the solution of the model for business cycles \eqref{51} induced by (\ref{3}).
The explicit representations we have derived help better understand the econometric behavior to the solutions of these models and
can serve as a starting point for further analysis of them.

\section*{Acknowledgments}

I would like to thank Professor L. Thomas Ramsey for motivating the study on model \eqref{3}, and Dr. Mingtao Lu for pointing out some references on models of economic growth and of business cycles. The major part of this article was completed when the author was at the Department of Mathematics, University of Hawaii, HI, USA.

\bibliographystyle{chicago}

\end{document}